\documentclass[12pt]{iopart}

\usepackage{iopams}  
\usepackage{graphicx}	
\usepackage{amssymb}	
\usepackage{amsthm}
\usepackage{tikz}
\usetikzlibrary{shapes.geometric, arrows}
\usepackage{tkz-euclide}

\def\d{\hspace{0.15em} \textup{d}}
\def\B{\mathbf{B}}

\def\A{\mathbf{A}}
\def\X{\mathbf{X}}
\def\Y{\mathbf{Y}}
\def\Q{\mathbf{Q}}
\def\d{\,\textup{d}}

\def\n{\mathbf{n}}
\def\t{\mathbf{t}}

\def\u{\mathbf{u}}

\def\x{\mathbf{x}}
\def\0{\mathbf{0}}

\def\brho{\boldsymbol\rho}

\newtheorem{theorem}{Theorem}

\newtheorem{definition}[theorem]{Definition}

\newtheorem{remark}[theorem]{Remark}
\def\curl{{\rm curl}}
\def\div{{\rm div}}

\begin{document}

\title[Relative helicity in multiply connected domains]{Relative magnetic helicity in multiply connected domains}

\author{David MacTaggart}

\address{School of Mathematics and Statistics, University of Glasgow, G12 8QQ, UK}
\ead{david.mactaggart@glasgow.ac.uk}
\author{Alberto Valli}
\address{Department of Mathematics, University of Trento, Povo, Italy}
\vspace{10pt}
\begin{indented}
\item[]July 2023
\end{indented}

\begin{abstract}
Magnetic helicity is a conserved quantity of ideal magnetohydrodynamics (MHD) that is related to the topology of the magnetic field, and is widely studied in both laboratory and astrophysical plasmas. When the magnetic field has a non-trivial normal component on the boundary of the domain, the classical definition of helicity must be replaced by \emph{relative magnetic helicity}. The purpose of this work is to review the various definitions of relative helicity and to show that they have a common origin - a general definition of relative helicity in multiply connected domains. We show that this general definition is both gauge-invariant and is conserved in time under ideal MHD, subject only to closed and line-tied boundary conditions. Other, more specific, formulae for relative helicity, that are used frequently in the literature, are shown to follow from the general expression by imposing extra conditions on the magnetic field or its vector potential. 
\end{abstract}

%
%
%
%
%

\section{Introduction}
Magnetic helicity is an invariant of ideal magnetohydrodynamics (MHD) that has found many applications in astrophysical and laboratory plasmas. After its discovery in the late 1950s by Woltjer \cite{woltjer1958}, it was quickly realized that magnetic helicity has a close connection to the linkage of the magnetic field \cite{moffatt1969} (see \cite{moreau1961} for a hydrodynamical version) and can, therefore, be considered a topological invariant of MHD. This original form of magnetic helicity is referred to as \emph{classical magnetic helicity}, and is concerned specifically with closed magnetic fields, i.e. those that are bounded by magnetic flux surfaces. 

The definition of classical magnetic helicity is as follows. A magnetic field vector $\B$\footnote{Technically, $\B$ is the magnetic induction, but it is common practice in MHD to refer to it as the magnetic field.} is divergence-free ($\div\B=0)$ and can be written in terms of a vector potential $\A$, $\B=\curl\A$. The classical form of magnetic helicity is written as
\begin{equation}\label{hel}
    H = \int_\Omega \A\cdot\B\, \d^3x,
\end{equation}
where $\B\cdot\n=0$ on $\partial\Omega$ and $\n$ is  outward normal unit vector on $\partial\Omega$. It is a standard result to show that $H$ is \emph{gauge-invariant}, i.e. independent of the choice of $\A$, if $\Omega$ is a simply connected domain (e.g. \cite{berger1984}).

Equation (\ref{hel}) has been used for both simply and multiply connected domains. For multiply connected domains, however, equation (\ref{hel}) is not gauge-invariant and its application in such domains requires the selection of a suitable gauge. A typical choice is the so-called Coulomb gauge, which has led to a deeper understanding of the relationship between helicity and the geometry and topology of magnetic fields (e.g. \cite{moffatt1969,moffatt1992}). 

Equation (\ref{hel}) was recently extended by MacTaggart and Valli \cite{mactaggart2019} to provide a gauge-invariant definition of magnetic helicity for multiply connected domains. The expression can be written as
\begin{equation}\label{hel_gen}
    H = \int_\Omega \A\cdot\B\, \d^3x -\sum_{i=1}^g\left(\oint_{\gamma_i}\A\cdot\t_i\,\d x\right)\left(\int_{\Sigma_i}\B\cdot\n_{\Sigma_i}\,\d^2x\right),
\end{equation}
where $g$ is the genus of $\Omega$, the $\gamma_i$ are closed paths around non-bounding surfaces in $\Omega$, the $\t_i$ are unit tangent vectors on the $\gamma_i$,  the $\Sigma_i$ are cut surfaces internal to $\Omega$ and the $\n_{\Sigma_i}$ are the unit normal vectors of the cut surfaces. More details of these geometric features will be reviewed later.

Equation (\ref{hel_gen}) succinctly generalizes the definition of magnetic helicity that is often given in the literature, where either
\[
\oint_{\gamma_i}\A\cdot\t_i\,\d x := 0 \quad {\rm or} \quad \int_{\Sigma_i}\B\cdot\n_{\Sigma_i}\,\d^2x :=0, \quad {\rm for}\,\, i=1,\dots,g.
\]
The former condition is often justified as removing the influence of magnetic field outside the domain linking with that inside (e.g. \cite{browning2014}). The latter condition imposes zero magnetic flux and this condition is too strong for many applications.

Everything we have described until now has been for magnetic fields that are everywhere tangent to the boundary, i.e. $\B\cdot\n = 0$ on $\partial\Omega$ (or $|\B|$ tends to zero at a suitable rate if $|\x|\rightarrow\infty$). The natural extension is to consider $\B\cdot\n=f\ne0$ on $\partial\Omega$\footnote{We also assume that $\int_{\Gamma_r} f\,\d^2x=0$ on all the connected components $\Gamma_r$ of the boundary in order to admit a vector potential for $\B$.}. Indeed, this situation is often more practical, and has received particularly close attention in solar physics (see \cite{Pariat2020,mactaggart2021} and references within). The main issue with having non-trivial normal components of the magnetic field on the boundary is that equations (\ref{hel}) or (\ref{hel_gen}) are no longer gauge-invariant. To resolve this, a new gauge-invariant form of helicity was introduced in the mid 1980s, called \emph{relative magnetic helicity}, with varying definitions proposed by different authors \cite{jensen1984,berger1984,finn1985}. 

It was recognized that in order to produce a gauge-invariant helicity with non-trival normal components on the boundary, a comparison of two different magnetic fields with the same normal boundary conditions is required. Here, we follow the thread of the seminal work of Berger and Field \cite{berger1984}, who describe a way to construct relative helicity in a simply connected domain.

\begin{figure}[h!]
    \centering
    \begin{tikzpicture}

        \draw [red,thick] plot [smooth,samples=200, tension=1] coordinates { 
   (-4,0) (-4.5,3) (-3,3.4) (-3,0)};

\filldraw [white] (-2.84,1.32) circle (1.9pt);

      \draw [red,thick] plot [smooth,samples=200, tension=1] coordinates { 
    (-2,0) (-2.5,1) (-3.5,2) (-2.5,3)  (-1,0)};




    \filldraw [white] (-2.861,3.18) circle (1.pt);

     \filldraw [white] (-2.968,3.16) circle (1.pt);
   
   \draw [blue,thick] plot [smooth,samples=200, tension=1] coordinates { 
   (-4,0) (-3.4,-1) (-3,0) };

     \draw [blue,thick] plot [smooth,samples=200, tension=1] coordinates { 
   (-2,0) (-1.4,-1) (-0.7,-1.3) (-1,0) };

\draw[thick,dashed]{} (0,-2) -- (0,0);


   \draw [blue,thick] plot [smooth,samples=200, tension=1] coordinates { 
   (-4+5,0) (-3.4+5,-1) (-3+5,0) };

     \draw [blue,thick] plot [smooth,samples=200, tension=1] coordinates { 
   (-2+5,0) (-1.4+5,-1) (-0.7+5,-1.3) (-1+5,0) };

     \draw [red,thick] plot [smooth,samples=200, tension=1] coordinates { 
   (-4+5,0) (-4.5+5,3) (-3+5,3.4) (-3+5,0) };

      \draw [red,thick] plot [smooth,samples=200, tension=1] coordinates { 
   (-2+5,0) (-2.2+5,1) (-1.7+5,1.3) (-1+5,0) };

     \node at (-2.5,4) [anchor=north]{$\B$};
     \node at (-2.5+5,4) [anchor=north]{$\B'$};

    \node at (-2.5,-1) [anchor=north]{$\widetilde{\B}$};
    \node at (2.5,-1) [anchor=north]{$\widetilde{\B}$};

    \node at (-4.7,0.5) [anchor=north]{$\Omega$};

    \node at (-4.7+5,0.5) [anchor=north]{$\Omega$};

     \draw[thick]{} (-5,0) -- (5,0);
     \draw[thick]{} (-5,0) -- (-5,4);
     \draw[thick]{} (-5,4) -- (5,4);
     \draw[thick]{} (5,4) -- (5,0);
     \draw[thick]{} (0,0) -- (0,4);
   
    \end{tikzpicture}
    \caption{A representation of the fields used in the construction of relative helicity. The domain $\Omega$ is shown twice - on the left containing the magnetic field $\B$, and on the right containing the reference magnetic field $\B'$. Both fields are closed outside $\Omega$ by the same $\widetilde{\B}$.}
    \label{diag_rel_hel}
\end{figure}
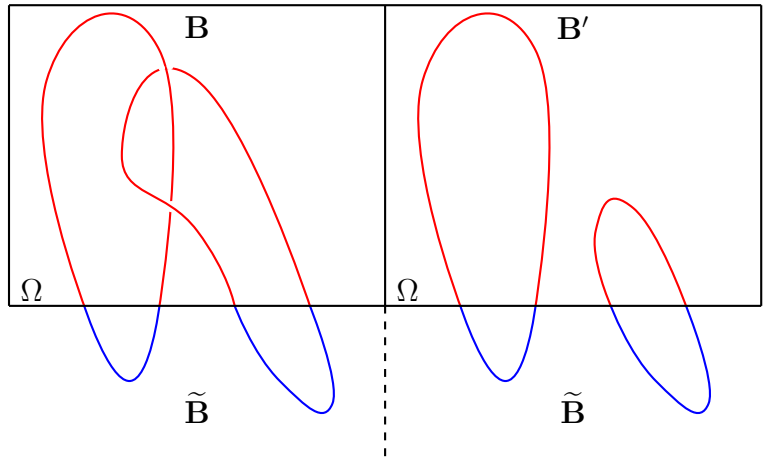

Consider a magnetic field $\B$ in a simply connected $\Omega$, as shown on the left-hand side of Figure \ref{diag_rel_hel}. The first step is to \lq close' the magnetic field (shown in red) by extending with another divergence-free field $\widetilde{\B}$ (shown in blue) outside of $\Omega$ in such a way that the extension matches the boundary components of $\B$ on $\partial\Omega$ ($\widetilde{\B}\cdot\n = 0$ on a \lq far' boundary). In this way $H(\B,\widetilde{\B})$ can be determined in the classical way via equation (\ref{hel}). However, this quantity is not useful by itself as it now extends outside the domain of interest and, further, it is not unique - $\widetilde{\B}$ is arbitrary up to boundary conditions. The solution lies in considering a second divergence-free \emph{reference} field in $\Omega$, $\B'$ say, that satisfies $\B'\cdot\n=\B\cdot\n = f$ on $\partial\Omega$. This is shown on the right-hand side of Figure \ref{diag_rel_hel}. Then, by adding the same extension $\widetilde{\B}$ to $\B'$, it can be shown that
\begin{equation}\label{hel_diff1}
    H^R_{1} = H(\B,\widetilde{\B})-H(\B',\widetilde{\B}),
\end{equation}
is gauge-invariant subject to the  condition that $\n\times\A=\n\times\A'$ on $\partial\Omega$. Further, equation (\ref{hel_diff1}) can be reduced to the simple form
\begin{equation}\label{hel_diff2}
    H^R_{1} = \int_\Omega(\A\cdot\B-\A'\cdot\B')\,\d^3x,
\end{equation}
which, importantly, is independent of the extension $\widetilde{\B}$. It can also be shown that $H^R_{1}$, as written above, is an invariant of ideal MHD for fixed $\Omega$ and closed and line-tied boundary conditions (e.g. \cite{laurence1991}).

Another definition of relative helicity that is popular in the literature is the so-called \emph{Finn-Antonsen formula} \cite{finn1985},
\begin{equation}\label{fa}
    H^R_2 = \int_\Omega(\A+\A')\cdot(\B-\B')\,\d^3 x.
\end{equation}
For simply connected $\Omega$, this expression of relative helicity is gauge-invariant for both $\A$ and $\A'$ independently, and without any condition on the tangential traces of the vector potentials. Although Finn and Antonsen \cite{finn1985} do not provide a clear geometrical construction for equation (\ref{fa}), as Berger and Field \cite{berger1984} did for equation (\ref{hel_diff2}), it is not difficult to show that $H^R_2$ becomes $H^R_1$ upon imposing $\n\times\A=\n\times\A'$ on $\partial\Omega$, namely

\begin{eqnarray*}
    H^R_2 &= \int_\Omega(\A\cdot\B-\A'\cdot\B')\,\d^3x + \int_\Omega(\A'\cdot\B-\A\cdot\B')\,\d^3x \\
    &= H^R_1 + \int_\Omega (\A'\cdot\curl\A - \A\cdot\curl\A')\,\d^3x\\
    & = H^R_1 + \int_{\partial\Omega}\A'\cdot(\n\times\A)\,\d^2x \\
    &= H^R_1 \,\, {\rm with} \,\, \n\times\A=\n\times\A' \,\, {\rm on} \,\, \partial\Omega. 
\end{eqnarray*}
Finn and Antonsen \cite{finn1985} also claim that $H_2^R$ is suitable for multiply connected domains subject to the conditions that the fluxes of $\B$ and $\B'$ are equal both inside and outside $\Omega$. They focus their discussion on toroidal domains, but provide no rigorous proof of their claims for $H_2^R$ in multiply connected domains.

The purpose of this work is to provide a general definition of relative helicity in multiply connected domains, from which the other formulae, including those for $H_1^R$ and $H_2^R$, are derived as special cases. In finding a suitable definition, we proceed with the following criteria: (1) the general definition should reduce to equation (\ref{hel_gen}) in the case where $\B\cdot\n=0$ everywhere on $\partial\Omega$, (2) the general definition must be gauge-invariant without extra conditions, such as those for $H^R_1$ and $H^R_2$, and must be conserved in time, for fixed $\Omega$, under ideal MHD, (3) the general definition should reduce to specific cases $H^R_1$ and $H^R_2$ by imposing extra conditions on either the magnetic field or its vector potential.

The rest of the work proceeds as follows. First, we present a summary of the geometrical setup and some fundamental results necessary for calculations in multiply connected domains (topological connections to these basic results are expanded upon in \ref{app_A}). This is followed by the main result, the general definition of relative helicity in multiply connected domains, for which we prove gauge-invariance and conservation in time under ideal MHD. We conclude the paper by remarking on how $H^R_1$ and $H^R_2$ follow from the general definition. 

\section{Preliminary results}

\subsection{Geometrical setup}

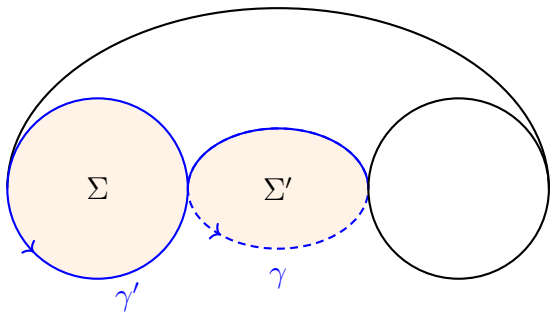
\begin{figure}[h!]
\centering
{\begin{tikzpicture}[scale=0.8]
\draw [fill=orange, opacity=0.1] (-3,-6) ellipse (1.5cm and 1.5cm);




\draw[black, thick]  (-4.5,-6) arc(-180:-360:4.5cm and 3cm);

\draw[draw=none,fill=orange,opacity=0.1]  (-1.5,-6) arc(-180:180:1.5cm and 1cm);

        
 \draw[
        blue,thick,densely dashed,decoration={markings, mark=at position 0.625 with {\arrow{>}}},
        postaction={decorate}
        ]
        (0,-6) ellipse (1.5cm and 1cm);


\draw[blue, thick]  (-1.5,-6) arc(-180:-360:1.5cm and 1cm);

\draw[
        blue,thick,decoration={markings, mark=at position 0.625 with {\arrow{>}}},
        postaction={decorate}
        ]
        (-3,-6) ellipse (1.5cm and 1.5cm);
\draw [thick] (3,-6) ellipse (1.5cm and 1.5cm);

\node at (0,-6) {$\Sigma'$};
\node at (-3,-6) {$\Sigma$};

\node[blue] at (-2.5,-7.8) {$\gamma'$};

\node[blue] at (0,-7.5) {$\gamma$};

\end{tikzpicture}}
\caption{A torus $\Omega$, shown cut in half to reveal the cross-section. The non-bounding cycles of $\Omega$, $\gamma$, and $\Omega'$, $\gamma'$, are shown as blue curves. The associated cutting surfaces, $\Sigma'$ and $\Sigma$, are indicated by orange surfaces. \label{torus}}
\end{figure}   
\unskip

Let $\Omega$ be a multiply connected domain that is a bounded open connected set with Lipschitz boundary $\partial\Omega$ and unit outer normal $\n$. Let $g>0$ be the first Betti number (or genus) of $\Omega$. This being the case, the first Betti number of $\partial\Omega$ is $2g$. We can consider $2g$ non-bounding cycles $\{\gamma_i\}^g_{i=1}\cup\{\gamma'_i\}^g_{i=1}$, that represent the generators of the first homology group of $\partial\Omega$. The $\{\gamma_i\}^g_{i=1}$ represent the generators of the first homology group of $\overline{\Omega}$. In Figure \ref{torus}, there is only $\gamma$ ($\equiv\gamma_1$), which encircles the hole of the torus. The tangent vector on $\gamma_i$ is denoted $\t_i$. Analogously, the $\{\gamma'_i\}^g_{i=1}$ represent the generators of the first homology group of $\overline{\Omega}'$, where $\Omega' = B\setminus\overline{\Omega}$, and $B$ is an open ball containing $\overline{\Omega}$. In Figure \ref{torus}, $\gamma'$ encircles the cross-section of the torus.  The tangent vector on $\gamma'_i$ is denoted $\t'_i$. 

There exist $g$ \lq cutting surfaces' $\{\Sigma_i\}_{i=1}^g$ in $\Omega$. These are connected orientable Lipschitz surfaces satisfying $\Sigma_i\subset\Omega$. Each cutting surface in $\Omega$ has boundary $\partial\Sigma_i=\gamma_i'$. In Figure \ref{torus}, $\Sigma$ is the cutting surface of the cross-section of the torus. Each $\overline{\Sigma}_i$ \lq cuts' the corresponding $\gamma_i$ cycle at only one point. The unit normal vector on a cutting surface is denoted $\n_{\Sigma_i}$, oriented as $\gamma_i$. With this choice just described, it holds that $\t_i'=\n_{\Sigma_i}\times\n$. Analogously, $\{\Sigma'_i\}_{i=1}^g$ are the cutting surfaces in $\Omega'$, with $\partial\Sigma_i'=\gamma_i$. In Figure \ref{torus}, $\Sigma'$ is the surface that spans the hole of the torus.

\subsection{Hodge decomposition}\label{sec_hodge}
We will make extensive use of the Hodge decomposition in multiply connected domains (e.g. \cite{blank1957,cantarella2002,alonso2010,ghiloni2010,alonso2018}). This states that any vector field $\Q\in L^2(\Omega)^3$ can be decomposed as
    \[
    \Q = \curl \A + \nabla\phi +\brho,
    \]
    where $\A$ is a vector field, $\phi$ is a scalar function and $\brho$ is a vector field in the space of Neumann harmonic fields, defined as
    \[
    \mathcal{H}(\Omega) = \{\brho \in L^2(\Omega)^3: \curl \brho=\boldsymbol{0},\, \div\brho=0\,\, {\rm in}\,\,\Omega;\,\,\brho\cdot\n=0\,{\rm on}\,\, \partial\Omega\}.
    \]
    Moreover, if $\curl\Q = \boldsymbol{0}$, then it follows that $\A=\boldsymbol{0}$. For a simply connected domain, $\mathcal{H}(\Omega)=\{\boldsymbol{0}\}$, so the topology of a multiply connected domain is encoded in $\mathcal{H}(\Omega)$. We now describe how to characterize the basis $\mathcal{H}(\Omega)$ in terms of the cycles and cutting surfaces introduced above. 
    
     We have that ${\rm dim}(\mathcal{H}(\Omega))=g$ and, following \cite{foias1978}, a basis $\{\brho\}_{i=1}^g$ can be constructed where $\brho_i=\widetilde{\nabla}\phi_i$, for $\phi\in H^1(\Omega\setminus\Sigma_i)$, with the following properties:
\begin{eqnarray*}
    \Delta \phi_i &= 0 \quad {\rm in} \quad \Omega\setminus\Sigma_i,\\
    \partial_{\n}\phi_i &= 0 \quad {\rm on} \quad \partial\Omega, \\
    {[\![\partial_{\n}\phi_i]\!]}_{\Sigma_i} &= 0, \\ 
    {[\![\phi_i]\!]}_{\Sigma_i} &= 1,
\end{eqnarray*}
where $ [\![\cdot]\!]_{\Sigma_i}$ denotes the jump across $\Sigma_i$. The notation $\widetilde{\nabla}\phi_i$ refers to the extension of this quantity to $L^2(\Omega)^3$. We discuss how this extension relates to the topology of $\Omega$ in \ref{app_A}. Analogous definitions apply for harmonic fields in $\Omega'$.

With this construction of the basis of $\mathcal{H}(\Omega)$, we can now perform calculations. Here, we are going to derive an orthogonality result that will be important later. For brevity, we will only outline the key steps and direct the reader to other work containing auxiliary derivations that are required for the derivation. In $\Omega$, let $\X$ be a divergence-free vector field, tangent to the boundary, and let $\Y$ be a curl-free vector field. Now, writing $\X=\curl\A$, for some vector potential $\A$, we have
\begin{equation}\label{orth1}
\int_\Omega\X\cdot\Y\,\d^3x = \int_\Omega\curl\A\cdot\Y\,\d^3x = \int_{\partial\Omega}\n\times\A\cdot\Y\,\d^2x,
\end{equation}
where the last integral follows from Green's identity with $\curl\Y=\boldsymbol{0}$. 

Now consider two fields $\A_1$ and $\A_2$ in $\Omega$ with $\curl\A_1\cdot\n=\curl\A_2\cdot\n=0$ on $\partial\Omega$. Alonso-Rodríguez et al. \cite{alonso2018}, by expressing the tangential traces of these fields in terms of the basis vectors of $\mathcal{H}(\Omega)$ and $\mathcal{H}(\Omega')$, derived the following relation
\begin{equation}
    \int_{\partial\Omega}\A_1\times\n\cdot\A_2\,\d^2x = \sum_{i=1}^g\alpha_i\mu_i - \sum_{i=1}^g\beta_i\delta_i, 
\end{equation}
where
\[
\alpha_i = \oint_{\gamma_i}\A_1\cdot\t_i\,\d x, \quad \beta_i = \oint_{\gamma'_i}\A_1\cdot\t'_i\,\d x,
\]
\[
\delta_i = \oint_{\gamma_i}\A_2\cdot\t_i\,\d x, \quad \mu_i = \oint_{\gamma'_i}\A_2\cdot\t'_i\,\d x.
\]
For our problem, we identify $\A_1:=\A$ and $\A_2:=\Y$. It is clear that 
\[
\beta_i = \int_{\Sigma_i}\X\cdot\n_{\Sigma_i}\,\d^2x, \quad \mu_i = 0,
\]
following the application of Stokes' theorem. Substituting these results into equation (\ref{orth1}) leads to the orthogonality relation
\begin{equation}\label{orth}
    \int_\Omega\X\cdot\Y\,\d^3x = \sum_{i=1}^g\left(\oint_{\gamma_i}\Y\cdot\t_i\,\d x\right)\left(\int_{\Sigma_i}\X\cdot\n_{\Sigma_i}\,\d^2x\right).
\end{equation}
This identity was used by MacTaggart and Valli \cite{mactaggart2019} to generalize classical helicity to multiply connected domains (equation (\ref{hel_gen})). We will also make use of equation (\ref{orth}) in generalizing relative helicity to multiply connected domains.

 Finally, we note that we will also make use of the fact that the basis functions satisfy 
\[
\oint_{\gamma_j}\brho_i\cdot\t_j\,\d x = \delta_{ij}, 
\]
for each $i$,$j=1,\dots,g$. We will expand upon this property in \ref{app_A}.
\section{General definition of relative helicity}

\begin{definition}\label{def}

    Let $\Omega\subset\mathbb{R}^3$ be a multiply connected domain as defined above. Consider divergence-free fields $\B$ and $\B'$ in $\Omega$ such that $\B\cdot\n=\B'\cdot\n=f$ on $\partial\Omega$. Denote the vector potentials of $\B$ and $\B'$ as $\A$ and $\A'$ respectively. The general definition of relative magnetic helicity, is given by
    \begin{eqnarray}\label{hr_gen}
        H^R = \int_\Omega(\A+\A')\cdot(\B-\B')\,\d^3x -\sum_{i=1}^g&\left(\oint_{\gamma_i}(\A+\A')\cdot\t_i\,\d x\right)\times\nonumber\\
        &\left(\int_{\Sigma_i}(\B-\B')\cdot\n_{\Sigma_i}\,\d^2x\right).
    \end{eqnarray}
\end{definition}

\begin{remark}
    It is clear that equation (\ref{hr_gen}) reduces to equation (\ref{hel_gen}) when $\B$ is everywhere tangent to $\partial\Omega$ and $\B'$ is set to zero. Thus, the classical helicity $H$ is contained within $H^R$.
\end{remark}

\begin{theorem}\label{theorem}
    $H^R$ given in Definition \ref{def} is gauge-invariant and, assuming $f\u=\mathbf{0}$ on $\partial\Omega$, is conserved in time under ideal MHD.
\end{theorem}

\begin{proof}
    We split the proof into two parts: gauge-invariance and time-conservation.

    \subsection{Part 1: gauge-invariance}
    The gauge-invariance of relative helicity follows primarily from the fact that the definition is based on $\B-\B'$, which is everywhere tangent to $\partial\Omega$. We now show that the particular combination $\A+\A'$ (which appears in both $H^R_2$ and $H^R$) is not important for gauge-invariance, and that any linear combination of $\A$ and $\A'$ will do. We thus replace $\A+\A'$ in Definition \ref{def} by $\mathcal{L}(\A,\A') = \alpha\A+\beta\A'$, for $\alpha$, $\beta\in\mathbb{R}$, and denote this modified form of the helicity by $H^R_{\mathcal{L}}$. Clearly, $H^R_{\mathcal{L}}$ reduces to $H^R$ when $\alpha=\beta=1$.

    Suppose we have two vector potentials $\A$ and $\A_\star$ of $\B$ and two vector potentials $\A'$ and $\A'_\star$ of $\B'$. Then

\begin{eqnarray*}
&H^R_{\mathcal{L}}(\A,\A') - H^R_{\mathcal{L}}(\A_\star,\A'_\star) \\ 
&\quad= \int_\Omega [\alpha({\A}-\A_\star) + \beta({\A}'-\A'_\star)] \cdot ({\B} - {\B}') \, \d^3x \\
&\quad- \sum_{i=1}^g \left(\oint_{\gamma_i} [\alpha({\A}-\A_\star) + \beta({\A}'-\A'_\star)]\cdot {\bf t}_i\, \d x\right) \left(\int_{\Sigma_i} (\B - \B') \cdot \n_{\Sigma_i}\,\d^2 x\right)\\
&\quad= 0.
\end{eqnarray*}
This result follows as since $\curl [\alpha({\A}-\A_\star) + \beta({\A}'-\A'_\star)] = \0$ in $\Omega$, $\div (\B - \B') = 0$ in $\Omega$ and $(\B - \B') \cdot \n =0$ on $\partial \Omega$, we can make use of the orthogonality relation (\ref{orth}). To summarize, $H^R_{\mathcal{L}}$ is gauge-invariant without any additional assumptions on $\B$ or $\B'$ (or their vector potentials).

\subsection{Part 2: conservation in time}
Laurence and Avellaneda \cite{laurence1991} showed that $H^R_1$ could be extended to multiply connected domains subject to a suitable choice of $\B'$ and gauge conditions. In the following, we assume no gauge conditions or particular choices of $\B'$, and the only constraint will be that of imposing closed and line-tied boundary conditions. In other words, as in Laurence and Avellaneda \cite{laurence1991}, we adopt
\begin{eqnarray*}
    \B\cdot\n &= \B'\cdot\n = f, \,\, \mbox{given and fixed in time},\\
    \u\cdot\n &=0, \,\, {\rm if}\,\,\B\cdot\n=0,\\
    \u &= \mathbf{0}, \,\, {\rm if}\,\,\B\cdot\n\ne0.
\end{eqnarray*}
The above boundary conditions for the velocity field can be encapsulated in the expression $f\u=\mathbf{0}$, and, as we will show, this expression is a crucial condition for $H^R$ to be conserved in time under ideal MHD if no extra conditions are imposed on the magnetic field. We note, in passing, that the boundary condition $f\u=\boldsymbol{0}$ is also assumed, among others, in Finn and Antonsen \cite{finn1985}.

To begin with, we will consider the time variation of $H_{\mathcal{L}}^R$ and show that this must reduce to $H^R$. With $\B'$ and $\A'$ being fixed in time, we have

\begin{eqnarray*}
\frac{\d H_{\mathcal{L}}^R(t)}{\d t} &= \alpha\int_\Omega \frac{\partial {\A}}{\partial t}\cdot ({\B} - {\B}') \, \d^3x + \int_\Omega  (\alpha{\A} + \beta\A')\cdot \frac{\partial {\B}}{\partial t}  \, \d^3x \\
&\qquad - \alpha\sum_{i=1}^g \Big(\oint_{\gamma_i}  \frac{\partial {\A}}{\partial t}\cdot {\bf t}_i dx\Big) \Big(\int_{\Sigma_i} (\B - \B') \cdot \n_{\Sigma_i} \d^2 x\Big) \\ 
&\qquad - \sum_{i=1}^g \Big(\oint_{\gamma_i} (\alpha{\A}+ \beta\A')\cdot {\bf t}_i dx\Big) \Big(\int_{\Sigma_i}  \frac{\partial  \B}{\partial t} \cdot \n_{\Sigma_i} \d^2 x\Big).
\end{eqnarray*}

Let us evaluate the four terms on the right-hand side in turn. Since $\partial{\B}/\partial t = \curl({\bf u} \times \B)$ in $\Omega$, we have $\curl (\partial\A/\partial t - {\bf u} \times \B) = \0$ in $\Omega$, therefore
\[
\frac{\partial \A}{\partial t} = {\u} \times \B + \nabla \psi + \sum_{j=1}^g \delta_j \brho_j,
\]
for some scalar function $\psi$ and $\delta_j\in\mathbb{R}$ $(j=1,\dots, g)$. Using this result, the first term is equal to
\begin{eqnarray}
\alpha\int_\Omega \frac{\partial{\A}}{\partial t}\cdot ({\B} - {\B}') \, \d^3x &= \alpha\int_\Omega \Big({\bf u} \times \B + \nabla \psi + \sum_{j=1}^g \delta_j \brho_j\Big)\cdot ({\B} - {\B}') \, \d^3x  \nonumber\\ 
&= -\alpha\int_\Omega ({\u} \times \B) \cdot {\B}' \, \d^3x \nonumber\\
&\quad+ \alpha\sum_{j=1}^g \delta_j \int_{\Sigma_j} ({\B} - {\B}') \cdot \n_{\Sigma_j} \, \d^2x,\label{first}
\end{eqnarray}
having used that $\int_\Omega \nabla \psi \cdot ({\B} - {\B}') \, \d^3x = 0$ (e.g. \cite{berger1984}) and that $\int_\Omega \brho_j\cdot ({\B} - {\B}') \, \d^3x= \int_{\Sigma_j} ({\B} - {\B}') \cdot \n_{\Sigma_j} \, \d^2x$ (e.g. \cite{mactaggart2019}), since $\B-\B'$ is a divergence-free and tangential vector field.

The second term takes the value
\begin{eqnarray}
&\int_\Omega  (\alpha{\A} + \beta\A')\cdot \frac{\partial {\B}}{\partial t}  \, \d^3x = \int_\Omega  (\alpha{\A} + \beta\A')\cdot \curl ({\u} \times \B)  \, \d^3x \nonumber\\ 
&= \int_\Omega  (\alpha{\B} + \beta\B')\cdot ({\u} \times \B)  \, \d^3x + \int_{\partial \Omega} (\alpha{\A} + \beta\A')\cdot [\n \times ({\bf u} \times\B)] \, \d^2 x \nonumber\\
&= \beta\int_\Omega ({\u} \times \B) \cdot \B'  \, \d^3x+ \int_{\partial \Omega} (\alpha{\A} + \beta\A')\cdot [(\B \cdot \n){\u} - ({\u} \cdot \n)\B] \, \d^2 x\nonumber \\
&= \beta\int_\Omega ({\u} \times \B) \cdot \B'  \, \d^3x +\int_{\partial \Omega} ({\alpha\A} + \beta\A')\cdot f {\u}  \, \d^2 x.\label{second}
\end{eqnarray}

The third term is equal to
\begin{eqnarray*}
&- \alpha\sum_{i=1}^g \Big(\oint_{\gamma_i}  \frac{\partial {\A}}{\partial t}\cdot {\t}_i \d x\Big) \Big(\int_{\Sigma_i} (\B - \B') \cdot \n_{\Sigma_i} \d^2 x\Big)\nonumber \\ 
&\qquad = - \alpha\sum_{i=1}^g \Big(\oint_{\gamma_i}  ({\u} \times \B + \nabla \psi + \sum_{j=1}^g \delta_j \brho_j)\cdot {\t}_i \d x\Big) \times \nonumber\\
&\qquad\quad\Big(\int_{\Sigma_i} (\B - \B') \cdot \n_{\Sigma_i} \d^2 x\Big)\nonumber \\
&\qquad =  - \alpha\sum_{i=1}^g \Big(\oint_{\gamma_i}  ({\u} \times \B)\cdot {\t}_i \d x\Big) \Big(\int_{\Sigma_i} (\B - \B') \cdot \n_{\Sigma_i} \d^2 x\Big) \\
&\qquad \quad\, - \alpha\sum_{j=1}^g \delta_j  \int_{\Sigma_j} (\B - \B') \cdot \n_{\Sigma_j} \d^2 x,
\end{eqnarray*}
having used that $\oint_{\gamma_i} \nabla \psi \cdot {\bf t}_i \, dx = 0$ for each $i=1,\ldots,g$, and that $\oint_{\gamma_i} \brho_j\cdot {\bf t}_i dx = \delta_{ij}$ for each $i,j=1,\ldots,g$. 

Observe that on the boundary $\partial \Omega$ we have $\B = \n \times \B \times \n + f \n$ and that ${\bf t}_i \times {\bf u}$ is a normal vector, therefore
\[
({\u} \times \B)\cdot {\t}_i = ({\t}_i \times {\u}) \cdot \B = ({\t}_i \times {\u})\cdot f \n = (\n \times {\t}_i) \cdot f{\u}. 
\]
Thus we conclude that the third term is given by
\begin{eqnarray}
&-\alpha\sum_{i=1}^g \Big(\oint_{\gamma_i}  \frac{\partial {\A}}{\partial t}\cdot {\t}_i \d x\Big) \Big(\int_{\Sigma_i} (\B - \B') \cdot \n_{\Sigma_i} \d^2 x\Big) \nonumber\\
&\qquad =  - \alpha\sum_{i=1}^g \Big(\oint_{\gamma_i}  (\n \times {\t}_i) \cdot f {\u} \, \d x\Big) \Big(\int_{\Sigma_i} (\B - \B') \cdot \n_{\Sigma_i} \d^2 x\Big)\nonumber \\
&\qquad\quad\, - \alpha\sum_{j=1}^g \delta_j  \int_{\Sigma_j} (\B - \B') \cdot \n_{\Sigma_j} \d^2 x,\label{third}
\end{eqnarray}
Before tackling the fourth term, we require the preliminary result,
\begin{eqnarray*}
\int_\Omega \B \cdot \brho_i \, \d^3x &= \int_\Omega \B \cdot  \widetilde \nabla \phi_i \, \d^3 x = \int_{\Omega\setminus\Sigma_i} \B \cdot \nabla \phi_i \, \d^3 x\\
&= \int_{\partial\Omega\setminus\partial\Sigma_i} (\B\cdot\n)\phi_i\,\d^2x + \int_{\Sigma_i}\B\cdot\n_{\Sigma_i}\,\d^2x\\
&=\int_{\partial \Omega} f \, \phi_i \, \d^2 x + \int_{\Sigma_i} \B \cdot \n_{\Sigma_i} \, \d^2 x.  
\end{eqnarray*}
Now, with $f$ being given and fixed in time,

\begin{eqnarray*}
\int_{\Sigma_i} \frac{\partial \B}{\partial t} \cdot \n_{\Sigma_i} \, \d^2 x &= \frac{\d}{\d t} \int_{\Sigma_i} \B \cdot \n_{\Sigma_i} \, \d^2 x = \frac{\d}{\d t} \int_\Omega \B \cdot \brho_i \, \d^3x\\
& = \int_\Omega \frac{\partial\B}{\partial t} \cdot \brho_i \, \d^3x =  \int_\Omega \curl ({\u} \times \B) \cdot \brho_i \, \d^3x\\
&= \int_{\partial \Omega} [\n \times ({\bf u} \times \B)]  \cdot \brho_i \, \d^2 x = \int_{\partial \Omega} f {\u} \cdot \brho_i \, \d^2 x,
\end{eqnarray*}
as $\n \times ({\u} \times \B) = f {\u}$. Therefore, the fourth term is given by
\begin{eqnarray}
&- \sum_{i=1}^g \Big(\oint_{\gamma_i} (\alpha{\A}+ \beta\A')\cdot {\bf t}_i \d x\Big) \Big(\int_{\Sigma_i}  \frac{\partial \B}{\partial t} \cdot \n_{\Sigma_i} \d^2 x\Big)\nonumber \\
& \qquad = - \sum_{i=1}^g \Big(\oint_{\gamma_i} (\alpha{\A}+ \beta\A')\cdot {\bf t}_i \d x\Big) \Big(\int_{\partial \Omega} f {\u} \cdot \brho_i \, \d^2 x\Big).\label{fourth}
\end{eqnarray}
Adding equations (\ref{first}-\ref{fourth}), we have
\begin{eqnarray}\label{dHdt}
    \frac{\d H_{\mathcal{L}}^R(t)}{\d t} &= (\beta-\alpha)\int_\Omega ({\u} \times \B) \cdot {\B}' \, \d^3x +\int_{\partial \Omega} ({\alpha\A} + \beta\A')\cdot f {\u}  \, \d^2 x \nonumber\\
    &\quad - \alpha\sum_{i=1}^g \Big(\oint_{\gamma_i}  (\n \times {\t}_i) \cdot f {\u} \, \d x\Big) \Big(\int_{\Sigma_i} (\B - \B') \cdot \n_{\Sigma_i} \d^2 x\Big) \nonumber\\
    &\quad - \sum_{i=1}^g \Big(\oint_{\gamma_i} (\alpha{\A}+ \beta\A')\cdot {\bf t}_i \d x\Big) \Big(\int_{\partial \Omega} f {\u} \cdot \brho_i \, \d^2 x\Big).
\end{eqnarray}
To eliminate the first term on the right-hand side of equation (\ref{dHdt}), we set $\alpha=\beta$. Since the resulting coefficient becomes just a scaling factor, we take $\alpha=\beta=1$, and so $H^R_{\mathcal{L}}$ reduces to $H^R$. To eliminate the remaining terms on the right-hand side of equation (\ref{dHdt}), we set $f\u=\mathbf{0}$ on $\partial\Omega$, and so
\[
\frac{\d H^R}{\d t} = 0.
\]
This completes the proof of Theorem \ref{theorem}.
\end{proof}

\section{Concluding remarks}
In this work, we have generalized relative magnetic relative helicity to multiply connected domains of very general topology. We have shown that equation (\ref{hr_gen}) satisfies criteria (1) and (2) from the end of the Introduction, namely that the general definition reduces to the classical case when $\B$ is tangent to $\partial\Omega$ (and $\B'$ can be ignored) and that the definition is gauge-invariant and conserved in time under ideal MHD (subject to natural boundary conditions). Further, there are no conditions imposed on $\B$ or $\B'$ (or their vector potentials) in equation (\ref{hr_gen}), beyond the standard magnetic boundary conditions $\B\cdot\n=\B'\cdot\n=f\ne0$.  

In reference to criterion (3), notice that $H^R$ reduces to $H^R_2$, the Finn-Antonsen formula (\ref{fa}), upon the imposition of the condition
\begin{equation}\label{flux_balance}
\int_{\Sigma_i} (\B - \B') \cdot \n_{\Sigma_i} \d^2 x=0,
\end{equation}
for $i=1,\dots,g$. It is interesting to note that although the more general expression, $H^R_{\mathcal{L}}$, is gauge-invariant, this reduces to $H^R$ when imposing conservation in time, and both gauge-invariance and time-conservation are needed for any quantity to be considered an invariant of ideal MHD. This provides an explanation for why the integrand of the Finn-Antonsen formula has its particular form. Finally, for the sake of completeness, we note that, as discussed earlier, $H^R_1$ follows from $H^R$ by imposing both equation (\ref{flux_balance}) and the condition $\A\times\n=\A'\times\n$ on $\partial\Omega$. We note that (\ref{flux_balance}) is necessary for this last condition, and the proof is given in \ref{app_B}.

\appendix

\section{}\label{app_A}
Here we elucidate on the connections between algebraic topology and vector fields in multiply connected domains and, in particular, the characterization of the space of harmonic fields. This material is not normally presented in a style that is immediately recognizable for a fluid dynamics/plasma physics audience, something which we hope to remedy shortly. First, however, we list some other works that would be useful for those interested in studying the relationship of topology to fluid dynamics and MHD.

A very readable account of the application of homology to MHD is given in Blank, Friedrichs and Grad \cite{blank1957}. The orthogonality relation (\ref{orth}) is derived in this work by a different method, but there is no discussion of helicity as the work predates the introduction of helicity by Woltjer \cite{woltjer1958}! Cantarella, DeTurck and Gluck \cite{cantarella2002} provide a clear introduction to the Hodge decomposition in $\mathbb{R}^3$, and this is a good starting place for readers new to this subject. A more technical approach is provided in the monograph by Alonso Rodríguez  and Valli \cite{alonso2010}, which focusses on computing electromagnetic fields in multiply connected domains.

In this paper, we have `cut' our domains in order to perform calculations, and this is a typical approach (e.g. \cite{foias1978,alonso2010,gross2004}). However, cuts need not be introduced \emph{a priori} to characterize the space of harmonic fields $\mathcal{H}(\Omega)$, and we will now follow the general approach of Ghiloni \cite{ghiloni2010} (in a much more condensed form). First, let us begin with some fundamental results.
\begin{theorem}[de Rham]
    The first homology group and the first de Rham cohomology group are finitely generated and have the same rank, that is given by the genus $g$ (equivalent to the first Betti number of $\overline{\Omega}$).
\end{theorem}
This theorem describes two groups. The first homology group is generated by $g$ (equivalence classes) of non-bounding cycles in $\overline{\Omega}$. In this work, we have labelled these cycles $\{\gamma_i\}_{i=1}^g$. Their equivalence classes $\{[\gamma_i]\}_{i=1}^g$ are generators of the first homology group. We speak of equivalence classes because deformations of the $\gamma_i$, that do not change the topology of the curves, have no effect (these are topological objects).

The dual description of homology is cohomology, and the first de Rham cohomology group is generated by (equivalence classes of) \emph{loop fields} in $\Omega$, which are curl-free vector fields that cannot be written as gradients in $\Omega$. It is these loop fields that we will use to construct the basis of $\mathcal{H}(\Omega)$. These fields also provide the extension of gradients to $L(\Omega)^3$, which we mentioned in Section \ref{sec_hodge}.

\begin{theorem}\label{lin_ind}
    A set of generators of the first de Rham cohomology group is given by the equivalence classes of $g$ loop fields $\hat{\brho}_i$ such that
    \[
    \oint_{\gamma_j}\hat{\brho}_i\cdot\t_j\,\d x = \delta_{ij},
    \]
    where $\delta_{ij}$ is the Kronecker delta.
\end{theorem}

\begin{proof}
    It is enough to show that the $\hat{\brho}_i$ are linearly independent. Suppose that we have
    \[
    \sum_i \alpha_i[\hat{\brho}_i] = \boldsymbol{0},
    \]
    for $\alpha_i\in \mathbb{R}$, which is equivalent writing
    \[
    \sum_i \alpha_i\hat{\brho}_i = \nabla\chi,
    \]
    for some scalar function $\chi$. It is clear that integrating on $\gamma_j$ leads to 
    \[
    0 = \oint_{\gamma_j}\sum_i\alpha_i\hat{\brho}_i\cdot\t_j\,\d x = \alpha_i\delta_{ij} = \alpha_j. 
    \]
\end{proof}
Although we will use these loop fields to characterize $\mathcal{H}(\Omega)$, they are not yet in a suitable form, i.e. although they are curl-free, they need to also be made divergence-free and everywhere tangent to the boundary. Let us denote $\omega_i$ to be the solution of the Neumann problem
\begin{eqnarray*}
    \Delta\omega_i &= \div\hat{\brho}_i \quad {\rm in}\,\,\, \Omega, \\
    \nabla\omega_i\cdot\n &= \hat{\brho}_i\cdot\n \quad {\rm on}\,\,\, \partial\Omega.
\end{eqnarray*}
This construction allows us to find a suitable projection of the loop fields to be divergence-free and everywhere tangent to the boundary.

\begin{theorem}
    The space $\mathcal{H}(\Omega)$ is finite dimensional with dimension $g$. A basis is given by $\brho_i = \hat{\brho}_i - \nabla\omega_i$, for $i=1,\dots, g$.
\end{theorem}

\begin{proof}
    In Theorem \ref{lin_ind}, we showed that the $\hat{\brho}_i$ are linearly independent. Since $\brho_i$ and $\hat{\brho}_i$ differ only by a gradient, it is clear that the $\brho_i$ must also be linearly independent.

    Let $\brho\in\mathcal{H}(\Omega)$. Its equivalence class $[\brho]$ is an element of the first de Rham cohomology group, hence we can write

    \[
    [\brho] = \sum_i\alpha_i[\hat{\brho}_i], \quad \brho = \sum_i\alpha_i\hat{\brho}_i + \nabla\chi,
    \]
    where $\chi$ satisfies
    \begin{eqnarray*}
        \Delta\chi &= -\sum_i\alpha_i\div\hat{\brho}_i = -\sum_i\alpha_i\Delta\omega_i\quad {\rm in}\,\,\,\Omega, \\
        \nabla\chi\cdot\n &= -\sum_i\alpha_i\hat{\brho}_i\cdot\n = -\sum_i\alpha_i\nabla\omega_i\cdot\n\quad{\rm on}\,\,\, \partial\Omega.
    \end{eqnarray*}
    Hence,
    \[
    \nabla\chi = -\sum_i\alpha_i\nabla\omega_i,
    \]
    and so 
    \[
    \brho = \sum_i\alpha_i\brho_i.
    \]
\end{proof}
This procedure shows how the basis of $\mathcal{H}(\Omega)$ connects directly to fundamental results in algebraic topology and how the topology of the domains that we have considered for helicity calculations is encoded in $\mathcal{H}(\Omega)$.

\section{}\label{app_B}

In multiply connected domains,  finding ${\A}'$ such that $\curl \A' = \B'$ in $\Omega$ with ${\A}' \times {\n} = {\A} \times {\n}$ on $\partial \Omega$ needs the compatibility conditions
\begin{equation}\label{comp1}
\div \B' = 0 \ \hbox{\rm in} \ \Omega \ \ , \ \ \div_\tau({\A} \times {\n}) = \B' \cdot \n \ \ \hbox{\rm on} \  \partial \Omega,
\end{equation}
where $\div_\tau$ is the surface divergence on $\partial\Omega$, and

\begin{equation}\label{comp2}
\int_\Omega \B' \cdot \brho_i \, d^3 x + \int_{\partial \Omega} ({\A} \times {\n}) \cdot \brho_i \, \d^2 x = 0 , \ \ i=1,\ldots,g.
\end{equation}
Conditions (\ref{comp1}) are satisfied as
\[
\div_\tau({\A} \times {\n}) = \curl \A \cdot \n = \B \cdot \n = f = \B' \cdot \n \ \ \hbox{\rm on} \ \partial \Omega.
\]
The first term in (\ref{comp2}) reads

\begin{eqnarray}
\int_\Omega \B' \cdot \brho_i \, \d^3 x &=\int_\Omega \B' \cdot \widetilde \nabla \phi_i \, \d^3 x = \int_{\Omega \setminus \Sigma_i} \B' \cdot  \nabla \phi_i \, \d^3 x \nonumber\\ 
&= - \int_{\Omega \setminus \Sigma_i} (\div \B') \phi_i \, \d^3 x + \int_{\partial \Omega \setminus \partial \Sigma_i} \B' \cdot \n \, \phi_i \,\ d^2 x\nonumber\\
&\quad+ \int_{\Sigma_i} \B' \cdot \n_{\Sigma_i} \, \d^2 x \nonumber\\ 
& = \int_{\partial \Omega} f \, \phi_i \, \d^2 x + \int_{\Sigma_i} \B' \cdot \n_{\Sigma_i} \, \d^2 x.\label{term1} 
\end{eqnarray}
The second term in (\ref{comp2}) reads

\begin{eqnarray}
\int_{\partial\Omega} ({\A} \times {\n})\cdot \brho_i\, \d^2 x  &=  \int_{\partial \Omega} ({\A} \times {\n}) \cdot \widetilde \nabla \phi_i \, \d^2 x = \int_{\partial \Omega \setminus \gamma'_i} ({\A} \times {\n}) \cdot \nabla \phi_i \, \d^2 x \nonumber\\
&=-\int_{\partial \Omega \setminus \gamma'_i} \div_\tau({\A} \times {\n}) \, \phi_i \, \d^2x\nonumber\\
&\quad+ \oint_{\gamma'_i}  ({\A} \times {\n}) \cdot \n_{\Sigma_i}\, \d x \nonumber\\
&=-\int_{\partial \Omega} \curl {\A} \cdot {\n} \, \phi_i \, \d^2x  + \oint_{\gamma'_i}  (\n \times {\n}_{\Sigma_i}) \cdot \A\, \d x \nonumber\\
&= -\int_{\partial \Omega} {\B} \cdot {\n} \, \phi_i \, \d^2x  - \oint_{\gamma'_i}  {\bf t}_i \cdot \A\, \d x \nonumber\\
&= -\int_{\partial \Omega} f \, \phi_i \, \d^2x  - \int_{\Sigma_i}  \curl \A \cdot \n_{\Sigma_i}
\, \d^2x \nonumber\\ 
&= -\int_{\partial \Omega} f \, \phi_i \, \d^2x  - \int_{\Sigma_i}  \B \cdot \n_{\Sigma_i} \, \d^2x. \label{term2}
\end{eqnarray}
Inserting (\ref{term1}) and (\ref{term2}) into (\ref{comp2}) reveals that the necessary conditions to be satisfied in the choice of $\B'$, for $H^R_1$ in multiply connected domains, are
\[
\int_{\Sigma_i} \B' \cdot \n_{\Sigma_i} \, \d^2 x = \int_{\Sigma_i} \B \cdot \n_{\Sigma_i} \, \d^2 x, 
\]
for $ i=1,\ldots,g.$ 
\section*{References}

\bibliography{rel_hel_bib}

\end{document}